\documentclass[a4paper,11pt]{article}
\usepackage[top=30truemm,bottom=30truemm,left=25truemm,right=25truemm]{geometry}

\usepackage{amsthm,mathtools,bm,color}
\usepackage{tikz}
\usetikzlibrary{matrix}

\usepackage{graphicx}
\usepackage[small]{caption}

\usepackage{hyperref}

\theoremstyle{plain}
\newtheorem{theorem}{Theorem}[section]
\newtheorem{lemma}[theorem]{Lemma}

\newtheorem{proposition}[theorem]{Proposition}

\newcommand{\yadd}[1]{\textcolor{black}{#1}}

\begin{document}

\title{Matrix similarity transformations derived from extended $q$-analogues of the Toda equation and Lotka-Volterra system}
%
%
\author{R. Watanabe\footnote{Faculty of Life and Environmental Sciences, Kyoto Prefectural University (r\_watanabe@mei.kpu.ac.jp)}, \and M. Shinjo\footnote{Faculty of Science and Engineering, Doshisha University (mshinjo@mail.doshisha.ac.jp)} \and and \and M. Iwasaki\footnote{Faculty of Life and Environmental Sciences, Kyoto Prefectural University (imasa@kpu.ac.jp)}}
\maketitle
\begin{abstract} 
The $q$-Toda equation is derived from replacing ordinary derivatives with $q$-derivatives in the famous Toda equation. 
In this paper, we associate an extension of the $q$-Toda equation with matrix eigenvalue problems, and then show applications of its time-discretization to computing matrix eigenvalues. 
With respect to the Lotka-Volterra system, we also have the similar discussion on the case of the Toda equation.
\end{abstract}
%
%
%
\section{Introduction}
\label{sec:1}
One of the most famous integrable systems is the \yadd{finite} Toda equation \cite{Flaschka_1974_1, Flaschka_1974_2} expressed using the so-called Flaschka's variables: 
\begin{align}
\label{eqn:Toda}
\left\{\begin{aligned}
& \dfrac{dx_k(t)}{dt}=y_k(t)-y_{k-1}(t),\quad k=1,2,\dots,m,\\ 
& \dfrac{dy_k(t)}{dt}=y_k(t)\left(x_{k+1}(t)-x_k(t)\right),\quad k=1,2,\dots,m-1,\\
& y_0(t)\coloneqq 0,\quad y_m(t)\coloneqq 0, 
\end{aligned}\right.
\end{align}
where $t$ and $k$ respectively denote continuous time and the spatial index. 
The Toda equation \eqref{eqn:Toda} was first considered in a study of nonlinear spring dynamics
and the Toda variables $x_k(t)$ and $y_k(t)$ respectively correspond to the $k$th point-mass variation from its momentum and the equilibrium point at continuous-time $t$ \cite{Toda_1967}.
The Toda equation \eqref{eqn:Toda} was also studied from the perspective of 
orthogonal polynomials \cite{Toda_1989} and rational approximations \cite{Nikishin_1991}.
Symes~\cite{Symes_1982} found an interesting relationship where the continuous-time evolution from $t$ to $t+1$ coincides with one step of the well-known $QR$ algorithm \cite{Golub_1996} in the case where target matrices are exponentials of symmetric tridiagonal matrices. 
A time-discretization \cite{Hirota_1981} of the Toda equation \eqref{eqn:Toda} leads to 
\begin{align}
\label{eqn:qd}
\left\{\begin{aligned}
& Q_k^{(n+1)}+E_{k-1}^{(n+1)}=Q_k^{(n)}+E_k^{(n)},\quad k=1,2,\dots,m,\\
& Q_k^{(n+1)}E_k^{(n+1)}=Q_{k+1}^{(n)}E_k^{(n)},\quad k=1,2,\dots,m-1,\\
& E_0^{(n)}\coloneqq 0,\quad E_m^{(n)}\coloneqq 0,
\end{aligned}\right.
\end{align}
where $Q_k^{(n)}$ and $E_k^{(n)}$ respectively denote the values of $Q_k$ and $E_k$ at discrete-time $n$.
Sogo also showed that the discrete Toda equation \eqref{eqn:qd} is simply the recursion formula of the quotient-difference (qd) algorithm \cite{Rutishauser_1990} for computing symmetric tridiagonal matrices. 
Similarity transformations employed in the qd algorithm are $LR$ transformations that decompose symmetric tridiagonal matrices into the products of lower bidiagonal matrices $L^{(n)}$ and upper bidiagonal matrices $R^{(n)}$ and then generate new symmetric tridiagonal matrices as the products $R^{(n)}L^{(n)}$. 
In other words, discrete-time evolutions in the discrete Toda equation \eqref{eqn:qd} give a sequence of $LR$ transformations for symmetric tridiagonal matrices. 
See also \cite{Gautschi_2002} for the symmetric $LR$ transformations from the perspective of the orthogonal polynomials. 
\par
An extension of the Toda equation \eqref{eqn:Toda} is the \yadd{Kostant-Toda} equation \cite{Kostant_1979}, and the hungry Toda equation is a special case of the \yadd{Kostant-Toda} equation \cite{Shinjo_2020}. 
The $q$-analogue of the Toda equation \cite{Area_2018}, which involves the so-called $q$-parameter satisfying $0<q<1$, can also be regarded as an extension of the Toda equation \eqref{eqn:Toda}. 
This is because taking the limit $q\to 1$ in the $q$-Toda equation leads to the Toda equation \eqref{eqn:Toda}.
The \yadd{Kostant-Toda} equation including the hungry Toda equation is associated with matrix eigenvalue problems, 
and a time-discretization of the hungry Toda equation is shown to be applicable to computing eigenvalues of totally nonnegative matrices 
whose minors are all positive \cite{Fukuda_2013}. 
However, to the best of our knowledge, the relationships of the $q$-Toda equation and its time-discretization 
to eigenvalue problems have not yet been clarified. 
The main purpose of this paper is thus to associate the $q$-Toda equation -- or more precisely, its further extension -- 
with eigenvalue problems and then present applications of its time-discretization to computing eigenvalues. 
Constructing matrix representations called Lax representations of integrable systems is useful in finding the relationships to eigenvalue problems. 
In particular, in the case of discrete integrable systems, the Lax representations are usually related to $LR$ transformations directly. 
The key point of this paper is not to be too particular about $LR$ transformations 
in designing Lax representations of the extended $q$-Toda equation and its time-discretization. 
Based on the resulting Lax representations, we examine distinctive similarity transformations 
generated from the extended $q$-discrete Toda equation.
\par
An integrable system closely related to the Toda equation is the Lotka--Volterra (LV) system, 
which describes predator--prey interactions 
\begin{align}
\label{eqn:LV}
\left\{\begin{aligned}
& \dfrac{du_k(t)}{dt}=u_k(t)\left(u_{k+1}(t)-u_{k-1}(t)\right),\quad k=1,2,\dots,2m-1,\\
& u_0(t)\coloneqq 0,\quad u_{2m}(t)\coloneqq 0, 
\end{aligned} \right. 
\end{align}
where $k$ denotes species index and $u_{k}(t)$ corresponds to the number of the $k$th species at continuous-time $t$. 
The hungry Lotka--Volterra (hLV) system \cite{Bogoyavlenskii_1991,Itoh_1987} is an extension of the LV system \eqref{eqn:LV}. 
The Darboux transformations on band matrices were studied in \cite{Adler_2000,Rolania_2015}, and then enable us to relate the \yadd{Kostant-Toda} equation to the discrete Korteweg de Vries equation which is equivalent to the hLV system \cite{Rolania_2019}. 
See also \cite{Rolania_2009} for other extensions of the Toda equation \eqref{eqn:Toda} and LV system \eqref{eqn:LV}.
A time-discretization of the LV (dLV) system
\begin{align}
\label{eqn:dLV}
\left\{\begin{aligned}
& u_k^{(n+1)}=u_k^{(n)}\dfrac{1+u_{k+1}^{(n)}}{1+u_{k-1}^{(n+1)}},\quad k=1,2,\ldots,2m-1, \\
& u_0^{(n)}\coloneqq 0,\quad u_{2m}^{(n)}\coloneqq 0. 
\end{aligned}\right. 
\end{align}
also has a close relationship to the discrete Toda equation. 
The dLV system \eqref{eqn:dLV} can be obtained from a variable transformation called the B\"acklund transformation or the Miura transformation of the discrete Toda equation \eqref{eqn:qd}. 
We designed an algorithm for computing singular values of bidiagonal matrices, 
which is equivalent to eigenvalues of symmetric positive-definite tridiagonal matrices based on the dLV system~\cite{Iwasaki_2002}. 
Similarly to the Toda case, the LV case has not been studied from the viewpoint of eigenvalue problems. 
In this paper, we thus focus on applications of the $q$-LV case to computing eigenvalues.
\par
The remainder of this paper is organized as follows.
In Section \ref{sec:2}, we briefly explain the $q$-Toda equation and then find Lax representations of its extension in order to associate it with eigenvalue problems. 
Next, in Section \ref{sec:3}, we consider a time-discretization of the extended $q$-Toda equation and relate it to a sequence of similarity transformations of matrices. 
In Section \ref{sec:4}, we investigate the case of the extended $q$-LV system.
In Section \ref{sec:5}, we clarify a determinantal solution to the extended $q$-discrete Toda equation and its asymptotic behavior as discrete-time goes to infinity. 
We also present numerical examples for confirming the convergence to matrix eigenvalues.
Finally, in Section \ref{sec:6}, we give concluding remarks.
%
\section{Extended $q$-analogue of the Toda equation}
\label{sec:2}
%
In this section, we first consider an extension of the $q$-analogue of the Toda equation \eqref{eqn:Toda}.
Next, by considering two kinds of Lax representations, 
we associate the extended $q$-Toda equation with matrix eigenvalue problems. 
\par
The $q$-derivative of the function $f(t)$ 
is a $q$-analogue of the ordinary derivative $df(t)/dt$ defined by:
\begin{align}
\label{eqn:q_derivative}
D_qf(t)\coloneqq \left\{
\begin{aligned}
& \dfrac{f(t)-f(qt)}{(1-q)t},\quad t\neq 0,\\
& \lim_{n\to\infty}\dfrac{f(q^nt)-f(0)}{q^nt},\quad t=0, 
\end{aligned}\right.
\end{align}
where $0<q<1$. Area et al. \cite{Area_2018} presented a $q$-analogue of the Toda equation: 
\begin{align}
\label{eqn:qToda}
\left\{\begin{aligned}
& D_qx_k(t)=g_k(qt)-g_{k-1}(qt),\quad k=1,2,\dots,m,\\ 
& D_qy_k(t)=g_k(qt)\left(x_{k+1}(qt)-x_k(t)\right),\quad k=1,2,\dots,m-1,\\
& y_0(t)\coloneqq 0,\quad y_m(t)\coloneqq 0, 
\end{aligned}\right.
\end{align}
where $g_k(qt)$ are auxiliary variables given by:
\begin{align}
\label{eqn:qToda_g}
\left\{\begin{aligned}
& g_1(qt)=\dfrac{y_1(qt)}{1+(1-q)tx_1(qt)},\\
& g_k(qt)=\dfrac{y_k(qt)}{y_{k-1}(t)}g_{k-1}(qt),\quad k=2,3,\dots,m-1,\\
& g_0(qt)\coloneqq 0,\quad g_m(qt)\coloneqq 0. 
\end{aligned}\right.
\end{align}
Since $D_qx_k(t)\to dx_k(t)/dt$, $D_qy_k(t)\to dy_k(t)/dt$ and $g_k(qt)\to y_k(t) $ as $q \to 1$, 
the $q$-Toda equation \eqref{eqn:qToda} with $q\to 1$ leads to the Toda equation \eqref{eqn:Toda}. 
\par
Now we introduce an arbitrary integer $M$ into the $q$-Toda equation \eqref{eqn:qToda} as
\begin{align}
\label{eqn:extended_qToda}
\left\{\begin{aligned}
& D_qx_k(t)=g_k(qt)-g_{k-M-1}(qt),\quad k=1,2,\dots,M_m+M,\\ 
& D_qy_k(t)=g_k(qt)\left(x_{k+M+1}(qt)-x_k(t)\right),\quad k=1,2,\dots,M_m-1,\\
& y_{-M}(t)\coloneqq 0,\quad y_{-M+1}(t)\coloneqq 0,\quad\dots,\quad y_0(t)\coloneqq 0,\\
& y_{M_m}(t)\coloneqq 0,\quad y_{M_m+1}(t)\coloneqq 0,\quad\dots,\quad y_{M_m+M}(t)\coloneqq 0, 
\end{aligned}\right.
\end{align}
where $M_i\coloneqq (M+1)i-M$ and $g_k(qt)$ satisfy 
\begin{align}
\label{eqn:extended_qToda_g}
\left\{\begin{aligned}
& g_k(qt)=\dfrac{y_k(qt)}{1+(1-q)tx_k(qt)},\quad k=1,2,\dots,M+1,\\
& g_k(qt)=\dfrac{y_k(qt)}{y_{k-M-1}(t)}g_{k-M-1}(qt),\quad k=M+2,M+3,\dots,M_m-1,\\
& g_{-M}(qt)\coloneqq 0,\quad g_{-M+1}(qt)\coloneqq 0,\quad\dots,\quad g_0(qt)\coloneqq 0,\\ 
& g_{M_m}(qt)\coloneqq 0,\quad g_{M_m+1}(qt)\coloneqq 0,\quad\dots,\quad g_{M_m+M}(qt)\coloneqq 0.
\end{aligned}\right.
\end{align}
Equation \eqref{eqn:extended_qToda} with $M=0$ is simply the $q$-Toda equation \eqref{eqn:qToda}. 
Thus, we hereinafter refer to \eqref{eqn:extended_qToda} as the extended $q$-Toda equation. 
Similarly to the simple Toda equation \eqref{eqn:Toda}, we can regard that the extended $q$-Toda equation \eqref{eqn:extended_qToda} describes a system of nonlinear springs. The extended $q$-Toda case however differs from the Toda equation \eqref{eqn:Toda} in that the $k$th point-mass links with the $(k-M-1)$st and $(k+M+1)$st point-masses rather than the $(k-1)$st and $(k+1)$st point-masses.
\par
We prepare $(M_m+M)\times(M_m+M)$ matrices involving the extended $q$-Toda variables 
$x_k(t)$, $y_k(t)$, and $g_k(t)$ as:
\[\begin{aligned}
& A(t)\coloneqq\left(\begin{array}{ccccc}
x_1(t) & \overbrace{0\ \cdots\ 0}^M & 1 & & \\
\begin{array}{c} 0 \\ \vdots \\ 0 \end{array}
& \ddots & \ddots & \ddots & \\
y_1(t) & \ddots & & \ddots & 1 \\
 & \ddots & \ddots & \ddots & 
\begin{array}{c} 0 \\ \vdots \\ 0 \end{array} \\
 & & y_{M_m-1}(t) & \underbrace{0\ \cdots\ 0}_M & x_{M_m+M}(t) 
\end{array}\right),\\
& G(t)\coloneqq\left(\begin{array}{ccccc}
0 & & \\\begin{array}{c} 0 \\ \vdots \\ 0 \end{array} 
 & \ddots & & \\ 
g_1(t) & \ddots & \ddots & \\
 & \ddots & \ddots & \ddots \\
 & & g_{M_m-1}(t) & \underbrace{0\ \cdots\ 0}_M & 0
\end{array}\right). 
\end{aligned}\]
We hereinafter call $A(t)$ and $G(t)$ the $(M+1)$-tridiagonal and $(M+1)$-subdiagonal matrices, respectively. 
The following proposition then gives the Lax representation 
for the extended $q$-Toda equation \eqref{eqn:extended_qToda}. 
\begin{proposition}
\label{prop:Lax_extended_qToda}
For the $(M+1)$-tridiagonal matrices $A(t)$ and the $(M+1)$-subdiagonal matrices $G(t)$, it holds that
\begin{align}
\label{eqn:Lax_extended_qToda}
D_qA(t)=A(qt)G(qt)-G(qt)A(t). 
\end{align}
\end{proposition}
\begin{proof}
Observing the $(k,k)$, $(k+M+1,k)$ and $(k+2M+2,k)$ entries of $A(qt)G(qt)-G(qt)A(t)$, we easily derive 
$g_k(qt)-g_{k-M-1}(qt)$, $g_k(qt)x_{k+M+1}(qt)-g_k(qt)x_k(t)$, and 
$g_k(qt)y_{k+M+1}(qt)-g_{k+M+1}(qt)y_k(t)$, respectively. 
The other entries of $A(qt)G(qt)$ and $G(qt)A(t)$ are all $0$. 
From the first and second equations of the extended $q$-Toda equation \eqref{eqn:extended_qToda}, 
it is obvious that $g_k(qt)-g_{k-M-1}(qt)=D_qx_k(t)$ and $g_k(qt)x_{k+M+1}(qt)-g_k(qt)x_k(t)=D_qy_k(t)$. 
Moreover, from the second equation of \eqref{eqn:extended_qToda_g}, 
we see that $g_k(qt)y_{k+M+1}(qt)-g_{k+M+1}(qt)y_k(t)=0$. 
Taking into account that the $(k+2M+2,k)$ entries of $D_qA(t)$ are all 0, 
we thus have \eqref{eqn:Lax_extended_qToda}.
\end{proof}
For convenience, we here replace $x_{M_k+s}(t)$, $y_{M_k+s}(t)$, and $g_{M_k+s}(qt)$ 
with $x_{k,s}(t)$, $y_{k,s}(t)$, and $g_{k,s}(qt)$, respectively. 
We also obtain another Lax representation for the extended $q$-Toda equation \eqref{eqn:extended_qToda} in terms of $m\times m$ tridiagonal and $1$-subdiagonal matrices:
\[\begin{aligned}
& {\cal A}_s(t)\coloneqq\left(\begin{array}{cccc}
x_{1,s}(t) & 1 & & \\
y_{1,s}(t) & x_{2,s}(t) & \ddots & \\
 & \ddots & \ddots & 1 \\
 & & y_{m-1,s}(t) & x_{m,s}(t)
 \end{array} \right), \\
& {\cal G}_{s}(t)\coloneqq\left(\begin{array}{ccccc}
0 & & \\
g_{1,s}(t) & \ddots & \\
 & \ddots & \ddots \\
 & & g_{m-1,s}(t) & 0\ 
\end{array} \right). 
\end{aligned}\]
\begin{proposition}
\label{prop:Lax_extended_qToda_1}
For the block diagonal matrices
$ {\cal A}(t)\coloneqq {\rm diag}({\cal A}_0(t),{\cal A}_1(t),\dots$, ${\cal A}_M(t))$ and
$ {\cal G}(t)\coloneqq {\rm diag}({\cal G}_0(t),{\cal G}_1(t),\dots,{\cal G}_M(t))$, it holds that
\begin{align}
\label{eqn:Lax_extended_qToda_1_all}
D_q{\cal A}(t)={\cal A}(qt){\cal G}(qt)-{\cal G}(qt){\cal A}(t). 
\end{align}
\end{proposition}
\begin{proof}
Obviously, $D_q{\cal A}(t)$ and ${\cal A}(qt){\cal G}(qt)-{\cal G}(qt){\cal A}(t)$ are both block diagonal matrices. 
Observing the $s$th diagonal blocks of $D_q{\cal A}(t)$ and ${\cal A}(qt){\cal G}(qt)-{\cal G}(qt){\cal A}(t)$, we derive
\begin{align}
\label{eqn:Lax_extended_qToda_1}
D_q{\cal A}_s(t)={\cal A}_s(qt){\cal G}_s(qt)-{\cal G}_s(qt){\cal A}_s(t),\quad s=0,1,\dots,M.
\end{align}
It is worth noting here that ${\cal A}_s(t)$ and ${\cal G}_s(t)$ have the same form as $A(t)$ and $G(t)$, respectively. 
Thus, from Proposition \ref{prop:Lax_extended_qToda}, 
we immediately obtain equalities of nonzero entries in \eqref{eqn:Lax_extended_qToda_1}:
\[\left\{\begin{aligned}
& D_qx_{k,s}(t)=g_{k,s}(qt)-g_{k-1,s}(qt),\quad k=1,2,\dots,m,\quad s=0,1,\dots,M,\\ 
& D_q y_{k,s}(t)= g_{k,s}(qt)\left(x_{k+1,s}(qt)-x_{k,s}(t)\right),\quad k=1,2,\dots,m-1,\quad s=0,1,\dots,M, 
\end{aligned}\right.\]
which is equivalent to the extended $q$-Toda equation \eqref{eqn:extended_qToda} under the replacements
$x_{k,s}(t)=x_{M_k+s}(t)$, $y_{k,s}(t)=y_{M_k+s}(t)$, and $g_{k,s}(t)=g_{M_k+s}(t)$. 
Therefore we conclude that \eqref{eqn:Lax_extended_qToda_1_all} is also 
the Lax representation of the extended $q$-Toda equation \eqref{eqn:extended_qToda}. 
\end{proof}
We here denote an eigenvalue of the $(M+1)$-tridiagonal matrix $A(t)$ by $\lambda_k$ 
and the corresponding eigenvector by $\Phi_k(t)$. 
From the $q$-derivative of the equality $A(t)\Phi_k(t)=\lambda_k\Phi_k(t)$, we then obtain:
\[D_qA(t)\cdot\Phi_k(t)+A(qt)\cdot D_q\Phi_k(t)=\lambda_kD_q\Phi_k(t).\]
Assuming that $D_q\Phi_k(t)=-G(qt)\Phi_k(t)$, 
we derive the Lax representation \eqref{eqn:Lax_extended_qToda} again. 
As the so-called Lax pair for the extended $q$-Toda equation \eqref{eqn:extended_qToda}, we thus have:
\begin{align}
\label{eqn:Lax_pair_qToda}
\begin{cases}
A(t)\Phi_k(t)=\lambda_k\Phi_k(t),\quad k=1,2,\dots,M_m+M,\\
D_q\Phi_k(t)=-G(qt)\Phi_k(t),\quad k=1,2,\dots,M_m+M. 
\end{cases}
\end{align}
Similarly, we obtain the Lax pair associated with the Lax representation \eqref{eqn:Lax_extended_qToda_1}:
\begin{align}
\label{eqn:Lax_pair_qToda_1}
\begin{cases}
{\cal A}_s(t)\Psi_{k,s}(t)=\sigma_{k,s}\Psi_{k,s}(t),\quad k =1,2,\dots,m,\\
D_q\Psi_{k,s}(t)=-{\cal G}_s(qt)\Psi_{k,s}(t),\quad k=1,2,\dots,m,
\end{cases}
\end{align}
where $(\sigma_{k,s},\Psi_{k,s}(t))$ denote eigenpairs of the block diagonal matrices ${\cal A}_s(t)$.
According to Nagata et al.~\cite{Nagata_2016}, 
for the $(M+1)$-tridiagonal matrices $A(t)$ and the block diagonal matrices $\Psi(t)$, 
there exists a permutation matrix $P$ such that $P^{\top}A(t)P=\Psi(t)$. 
This implies that eigenvalues of $A(t)$ and $\Psi(t)$ are equal to each other.
Namely, every $\sigma_{k,s}$ coincides with one of $\lambda_1,\lambda_2,\dots,\lambda_{M_m+M}$.
%
\section{$q$-discretization and similarity transformations}
\label{sec:3}
In this section, we present a time-discretization of the extended $q$-Toda equation \eqref{eqn:extended_qToda}, 
and then relate it to similarity transformations of $M$-tridiagonal matrices and block diagonal matrices.
\par
We introduce new variables $t^{(0)},t^{(1)},\dots$ sequentially given by $t^{(n)}\coloneqq q^{-1} t^{(n-1)}$, where $t^{(0)}>0$. 
The extended $q$-Toda equation \eqref{eqn:extended_qToda} with $t=t^{(n)}$ then leads to:
\begin{align}
\label{eqn:extended_qdToda}
\left\{\begin{aligned}
& \dfrac{x_{k,s}^{(n+1)}-x_{k,s}^{(n)}}{(1-q)t^{(n)}}=g_{k,s}^{(n)}-g_{k-1,s}^{(n)},\quad k=1,2,\dots,m,\quad s=0,1,\dots,M,\\ 
& \dfrac{y_{k,s}^{(n+1)}-y_{k,s}^{(n)}}{(1-q)t^{(n)}}=g_{k,s}^{(n)}\left(x_{k+1,s}^{(n)}-x_{k,s}^{(n+1)}\right),\\
& \qquad k=1,2,\dots,m-1,\quad s=0,1,\dots,M,\\
& y_{0,s}^{(n)}\coloneqq 0,\quad y_{m,s}^{(n)}\coloneqq 0, 
\end{aligned}\right.
\end{align}
where $x_{k,s}^{(n+1)}\coloneqq x_{k,s}(t^{(n)})$, $y_{k,s}^{(n+1)}\coloneqq y_{k,s}(t^{(n)})$,
$g_{k,s}^{(n+1)}\coloneqq g_{k,s}(t^{(n)})$, and $g_{k,s}^{(n)}$ satisfy
\begin{align}
\label{eqn:extended_qdToda_g}
\left\{\begin{aligned}
& g_{1,s}^{(n)}=\dfrac{y_{1,s}^{(n)}}{1+(1-q)t^{(n)}x_{1,s}^{(n)}},\quad s=0,1,\dots,M,\\
& g_{k,s}^{(n)}=\dfrac{y_{k,s}^{(n)}}{y_{k-1,s}^{(n+1)}}g_{k-1,s}^{(n)},
\quad k=2,3,\dots,m-1,\quad s=0,1,\dots,M,\\
& g_{0,s}^{(n)}\coloneqq 0,\quad g_{m,s}^{(n)}\coloneqq 0.
\end{aligned}\right.
\end{align}
We hereinafter refer to \eqref{eqn:extended_qdToda} as the extended $q$-discrete Toda equation.
It is emphasized here that the values of $x_{k,s}^{(n+1)}$, $y_{k,s}^{(n+1)}$ and $g_{k,s}^{(n+1)}$ 
are uniquely computed from those of $x_{k,s}^{(n)}$, $y_{k,s}^{(n)}$ and $g_{k,s}^{(n)}$ 
in the extended $q$-discrete Toda equation \eqref{eqn:extended_qdToda}. 
Figure \ref{fig:diagram_time_evolution} shows the discrete-time evolution from $n$ to $n+1$ 
in the extended $q$-discrete Toda equation \eqref{eqn:extended_qdToda}.
\begin{figure}[tb]
\centering
\includegraphics[width=0.7\textwidth]{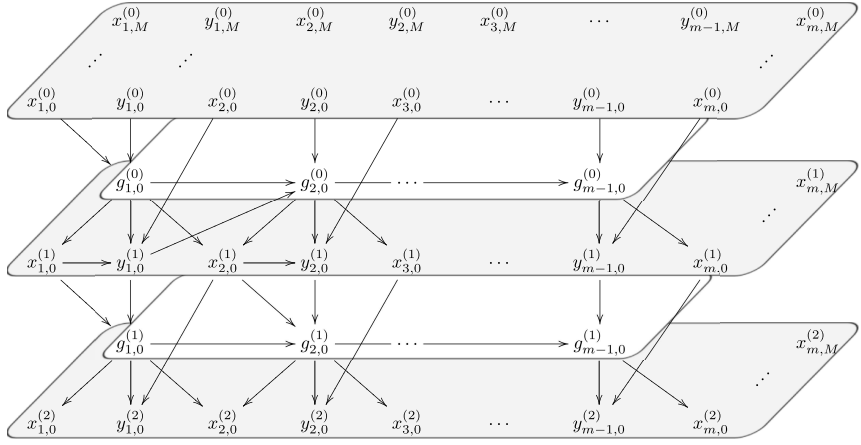}
\caption{Discrete-time evolution in the extended $q$-discrete Toda equation \eqref{eqn:extended_qdToda}. \label{fig:diagram_time_evolution}}
\end{figure}
\par
By setting $A^{(n+1)}\coloneqq A(t^{(n)}) $ and $G^{(n+1)}\coloneqq G(t^{(n)})$ 
in the Lax representation \eqref{eqn:Lax_extended_qToda}, 
we easily derive a Lax representation for the extended $q$-discrete Toda equation \eqref{eqn:extended_qdToda}.
\begin{theorem}
\label{thm:similar_transformation_qdToda}
For the $(M+1)$-tridiagonal matrices $A^{(n)}$ and the $(M+1)$-subdiagonal matrices $G^{(n)}$, it holds that
\begin{align*}
(I_{M_m+M}+(1-q)t^{(n)}G^{(n)})A^{(n+1)}=A^{(n)}(I_{M_m+M}+(1-q)t^{(n)}G^{(n)}), 
\end{align*}
where $I_{M_m+M}$ denotes the $(M_m+M)\times(M_m+M)$ identity matrix. 
\end{theorem}
\noindent
Since $\det(I_{M_m+M}+(1-q)t^{(n)}G^{(n)})=1$, the inverse of $I_{M_m+M}+(1-q)t^{(n)}G^{(n)}$ always exist.
Theorem \ref{thm:similar_transformation_qdToda} thus implies that 
the extended $q$-discrete Toda equation \eqref{eqn:extended_qdToda} 
generates a sequence of similarity transformations of the initial $A^{(0)}$. 
\par
Now we prepare auxiliary variables given using $x_{k,s}^{(n)}$, $y_{k,s}^{(n)}$, $g_{k,s}^{(n)}$, $t^{(n)}$, and $q$ as
\begin{align}
\label{eqn:auxiliary_variable_qdToda}
\left\{\begin{aligned}
& Q_{k,s}^{(n)}\coloneqq x_{k,s}^{(n)}+\dfrac{1}{(1-q)t^{(n)}}-(1-q)t^{(n)}g_{k-1,s}^{(n)},\\
& \qquad k=1,2,\dots,m,\quad s=0,1,\dots,M,\\
& E_{k,s}^{(n)}\coloneqq (1-q)t^{(n)}g_{k,s}^{(n)},\quad k=1,2,\dots,m,\quad s=0,1,\dots,M, 
\end{aligned}\right.
\end{align}
where $E_{0,s}^{(n)}\coloneqq 0$ and $E_{m,s}^{(n)}\coloneqq 0 $. 
The following proposition then describes discrete-time evolution from $n$ to $n+1$ of $Q_{k,s}^{(n)}$ and $E_{k,s}^{(n)}$. 
\begin{proposition}
\label{prop:auxiliary_recursion_formula_qdToda}
For $s=0,1,\dots,M$, it holds that
\begin{align}
\label{eqn:auxiliary_recursion_formula_qdToda}
\left\{\begin{aligned}
& Q_{k,s}^{(n+1)}+E_{k-1,s}^{(n+1)}-\dfrac{1}{(1-q)t^{(n+1)}}=Q_{k,s}^{(n)}+E_{k,s}^{(n)}-\dfrac{1}{(1-q)t^{(n)}},
 \quad k=1,2,\dots,m,\\
& Q_{k,s}^{(n+1)}E_{k,s}^{(n+1)}=Q_{k+1,s}^{(n)}E_{k,s}^{(n)},\quad k=1,2,\dots,m-1.
\end{aligned}\right.
\end{align}
\end{proposition}
\begin{proof}
Observing the first equation in the extended $q$-discrete Toda equation \eqref{eqn:extended_qdToda}, we obtain
\begin{align}
\label{eqn:auxiliary_recursion_formula_qdToda_Q}
\left\{\begin{aligned}
& Q_{k,s}^{(n)}=x_{k,s}^{(n)}-E_{k-1,s}^{(n)}+\dfrac{1}{(1-q)t^{(n)}}, \\
& Q_{k,s}^{(n)}=x_{k,s}^{(n+1)}-E_{k,s}^{(n)}+\dfrac{1}{(1-q)t^{(n)}}.
\end{aligned}\right.
\end{align}
Equation \eqref{eqn:auxiliary_recursion_formula_qdToda_Q} immediately 
leads to the first equation of \eqref{eqn:auxiliary_recursion_formula_qdToda}.
\par
By dividing both sides of the second equation of the extended $q$-discrete Toda equation \eqref{eqn:extended_qdToda} 
by $g_{k,s}^{(n)}$, recalling that $g_{k+1,s}^{(n)}y_{k,s}^{(n+1)}=g_{k,s}^{(n)}y_{k+1,s}^{(n)}$,
and using \eqref{eqn:auxiliary_variable_qdToda}, we derive:
\begin{align}
\label{proof:E_q_dToda_qd}
Q_{k+1,s}^{(n)}-\dfrac{y_{k+1,s}^{(n)}}{E_{k+1}^{(n,s)}}=Q_{k,s}^{(n)}-\dfrac{y_{k,s}^{(n)}}{E_{k,s}^{(n)}}. 
\end{align}
From \eqref{proof:E_q_dToda_qd}, it follows that:
\[Q_{k+1,s}^{(n)}-\dfrac{y_{k+1,s}^{(n)}}{E_{k+1}^{(n,s)}}=Q_{1,s}^{(n)}-\dfrac{y_{1,s}^{(n)}}{E_1^{(n,s)}}.\]
Since it is obvious from the first equation in \eqref{eqn:extended_qdToda_g} 
and \eqref{eqn:auxiliary_variable_qdToda} with $k=1$ that $Q_{1,s}^{(n)}-y_{1,s}^{(n)}/E_1^{(n,s)}=0$, we obtain:
\[Q_{k,s}^{(n)}-\dfrac{y_{k,s}^{(n)}}{E_{k,s}^{(n)}}=0.\] 
The second equation in \eqref{eqn:extended_qdToda_g} 
and the second equation in \eqref{eqn:auxiliary_variable_qdToda} also yield 
$y_{k,s}^{(n+1)}/E_{k+1,s}^{(n)}=y_{k+1,s}^{(n)}/E_{k,s}^{(n)}$. 
Thus, we have:
\begin{align}
\label{eqn:auxiliary_recursion_formula_qdToda_E}
\left\{\begin{aligned}
& Q_{k,s}^{(n)}-\dfrac{y_{k,s}^{(n)}}{E_{k,s}^{(n)}}=0,\\
& Q_{k,s}^{(n)}-\dfrac{y_{k-1,s}^{(n+1)}}{E_{k-1,s}^{(n)}}=0. 
\end{aligned}\right.
\end{align}
Equations \eqref{eqn:auxiliary_recursion_formula_qdToda_E} therefore 
lead to the second equation of \eqref{eqn:auxiliary_recursion_formula_qdToda}.
\end{proof}
\par
The following theorem also gives another matrix representation expressed using $(M_m+M)\times(M_m+M)$ lower and upper block bidiagonal matrices:
\[L^{(n)}\coloneqq\left(\begin{array}{cccc}
I_{M+1} & & & \\
L_1^{(n)} & I_{M+1} & & \\
 & \ddots & \ddots & \\
 & & L_{m-1}^{(n)} & I_{M+1} 
\end{array}\right),\quad
R^{(n)}\coloneqq\left(\begin{array}{cccc}
R_1^{(n)} & I_{M+1} & & \\
 & R_2^{(n)} & \ddots & \\
 & & \ddots & I_{M+1} \\
 & & & R_m^{(n)}
\end{array}\right),\]
where $L_k^{(n)}\coloneqq {\rm diag}(E_{k,0}^{(n)},E_{k,1}^{(n)},\dots,E_{k,M}^{(n)})$ 
and $R_k^{(n)}\coloneqq {\rm diag}(Q_{k,0}^{(n)},Q_{k,1}^{(n)},\dots,$ $Q_{k,M}^{(n)})$. 
\begin{theorem}
\label{thm:LR_transformation}
Under discrete-time evolution from $n$ to $n+1$ in \eqref{eqn:auxiliary_recursion_formula_qdToda}, it holds that
\begin{align}
\label{eqn:LU_extended_q_dToda}
L^{(n+1)}R^{(n+1)}-\dfrac{1}{(1-q)t^{(n+1)}}I_{M_m+M}=R^{(n)}L^{(n)}-\dfrac{1}{(1-q)t^{(n)}}I_{M_m+M}. 
\end{align}
\end{theorem}
\begin{proof}
All blocks in $L^{(n)}$ and $R^{(n)}$ are diagonal matrices of size $(M+1)\times(M+1)$, 
and their sums and products are also diagonal matrices. 
The $k$th diagonal block and lower off-diagonal block in $L^{(n+1)}R^{(n+1)}-\{1/[(1-q)t^{(n+1)}]\}I_{M_m+M}$ 
are, respectively, $R_k^{(n+1)}+L_{k-1}^{(n+1)}-\{1/[(1-q)t^{(n+1)}]\}I_m$ and $L_k^{(n+1)}R_k^{(n+1)}$. 
Moreover, the $(s+1,s+1)$ entries of the blocks $R_k^{(n+1)}+L_{k-1}^{(n+1)}-\{1/[(1-q)t^{(n+1)}]\}I_m$ 
and $L_k^{(n+1)}R_k^{(n+1)}$ are, respectively, $Q_{k,s}^{(n+1)}+E_{k-1,s}^{(n+1)}-1/[(1-q)t^{(n+1)}]$ 
and $Q_{k,s}^{(n+1)}E_{k,s}^{(n+1)}$, which are respectively equal to 
$Q_{k,s}^{(n)}+E_{k,s}^{(n)}-1/[(1-q)t^{(n)}]$ and $Q_{k+1,s}^{(n)}E_{k,s}^{(n)}$. 
Namely, the $(s+1,s+1)$ entries of the blocks $R_k^{(n+1)}+L_{k-1}^{(n+1)}-\{1/[(1-q)t^{(n+1)}]\}I_m$ 
and $L_k^{(n+1)}R_k^{(n+1)}$ coincide with the $(s+1,s+1)$ entries of the blocks 
$R_k^{(n)}+L_k^{(n)}-\{1/[(1-q)t^{(n)}]\}I_m$ and $R_{k+1}^{(n)}L_k^{(n)}$,
which are, respectively, the $k$th diagonal block and lower off-diagonal block 
in $R^{(n)} L^{(n)} - \{ 1/[(1-q)t^{(n)}] \}I_{M_m+M}$. 
Taking into account that the upper off-diagonal blocks of $L^{(n+1)}R^{(n+1)}-\{1/[(1-q)t^{(n+1)}]\}I_{M_m+M}$ and $R^{(n)}L^{(n)}-\{1/[(1-q)t^{(n)}]\}I_{M_m+M}$ are all $I_{M+1}$, 
we therefore have \eqref{eqn:LU_extended_q_dToda}.
\end{proof}
\noindent
Obviously, discrete-time evolution from $n$ to $n+1$ 
in the extended $q$-discrete Toda equation \eqref{eqn:extended_qdToda} 
completes that in \eqref{eqn:LU_extended_q_dToda}. 
Theorem \ref{thm:LR_transformation} thus suggests that 
the extended $q$-discrete Toda equation \eqref{eqn:extended_qdToda} implicitly generates 
similarity transformations, or more precisely, shifted $LR$ transformations of block tridiagonal matrices. 
%
\section{Extended $q$-discrete Lotka--Volterra system and similarity transformations}
\label{sec:4}
In this section, along the same lines as the case of the Toda equation \eqref{eqn:Toda},
we extend $q$-analogues of the LV system \eqref{eqn:LV}, and clarify the relationship to matrix similarity transformations.
\par 
Area et al. \cite{Area_2018} presented a $q$-analogue of the LV ($q$-LV) system \eqref{eqn:LV}:
\begin{align}
\label{eqn:qLV}
\left\{\begin{aligned}
& D_qu_k(t)=w_k(qt)\left(u_{k+1}(qt)-u_{k-1}(t)\right),\quad k=1,2,\dots,2m-1,\\
& u_0(t)\coloneqq 0,\quad u_{2m}(t)\coloneqq 0, 
\end{aligned} \right. 
\end{align}
where the auxiliary variables $w_k(qt)$ satisfy:
\begin{align}
\label{eqn:qLV_w}
\left\{\begin{aligned}
& w_1(qt)=\dfrac{u_1(qt)}{1+(1-q)tu_1(qt)},\\
& w_k(qt)=\dfrac{u_k(qt)}{u_{k-1}(t)}w_{k-1}(qt),\quad k=2,3,\dots,2m-1. 
\end{aligned}\right. 
\end{align}
Of course, the $q$-LV system \eqref{eqn:qLV} with $q\to 1$ immediately leads to the LV system \eqref{eqn:LV}. 
An extension of the $q$-LV system \eqref{eqn:qLV} is given involving arbitrary integer $M$ as:
\begin{align}
\label{eqn:extended_qLV}
\left\{\begin{aligned}
& D_qu_k(t)=w_k(qt)\left(u_{k+M+1}(qt)-u_{k-M-1}(t)\right),\quad k=1,2,\dots,M_{2m}-1,\\
& u_{-M}(t)\coloneqq 0,\quad u_{-M+1}(t)\coloneqq 0,\quad\ldots,\quad u_{0}(t)\coloneqq 0,\\
&u_{M_{2m}}(t)\coloneqq 0,\quad u_{M_{2m}+1}(t)\coloneqq 0,\quad\dots,\quad u_{M_{2m}+M}(t)\coloneqq 0, 
\end{aligned}\right. 
\end{align}
where $w_k(qt)$ satisfy:
\begin{align}
\label{eqn:extended_qLV_w}
\left\{\begin{aligned}
& w_k(qt)=\dfrac{u_k(qt)}{1+(1-q)tu_k(qt)},\quad k=1,2,\dots,M+1,\\
& w_k(qt)=\dfrac{u_k(qt)}{u_{k-M-1}(t)}w_{k-M-1}(qt),\quad k=M+2,M+3,\dots,M_{2m}-1. 
\end{aligned}\right. 
\end{align}
Obviously, \eqref{eqn:extended_qLV} with $M=0$ is the $q$-LV system \eqref{eqn:qLV}. 
We can regard that the extended $q$-LV system \eqref{eqn:extended_qLV} describes predator-prey interactions, assuming that the $k$th species preys on the $(k+M+1)$st species and is food for the $(k-M-1)$st species. In the simple LV case, the $k$th species preys on the $(k+1)$st species and is food for the $(k-1)$st species.
\par 
Given $(M_{2m}+M)\times(M_{2m}+M)$ matrices involving the extended $q$-LV variables $u_{k}(t)$ and $w_{k}(t)$:
\[\begin{aligned}
& U(t)\coloneqq\left(\begin{array}{ccccc}
0 & \overbrace{0\ \cdots\ 0}^M & 1 & & \\
\begin{array}{c} 0 \\ \vdots \\ 0 \end{array}
 & \ddots & \ddots & \ddots & \\
u_{M_1}(t) & \ddots & \ddots & \ddots & 1 \\
 & \ddots & \ddots & \ddots &
 \begin{array}{c} 0 \\ \vdots \\ 0 \end{array} \\
 & & u_{M_{2m}-1}(t) & \underbrace{0\ \cdots\ 0}_{M} & 0
\end{array}\right),\\
& W(t)\coloneqq\left(\begin{array}{ccccc}
0 & & & & \\
\begin{array}{c} 0 \\ \vdots \\ 0 \end{array} & \ddots & & & \\
w_{M_1}(t)u_{M_2}(t) & \ddots & \ddots & & \\
 & \ddots & \ddots & \ddots & \\
 & & w_{M_{2m}-M-2}(t)u_{M_{2m}-1}(t) & \underbrace{0\ \cdots\ 0}_{2M+1} & 0
\end{array}\right), 
\end{aligned}\]
we derive a Lax representation for the extended $q$-LV system \eqref{eqn:extended_qLV}.
\begin{proposition}
\label{prop:Lax_extended_qLV}
For $U(t)$ and $W(t)$, it holds that
\begin{align}
\label{eqn:Lax_extended_qLV}
D_qU(t)=U(qt)W(qt)-W(qt)U(t). 
\end{align}
\end{proposition}
\begin{proof}
The $(k+M+1,k)$ and $(k+2M+2,k)$ entries of $U(qt)W(qt)$ are, respectively, $w_{k}(qt)u_{k+M+1}(qt)$ and $w_{k}(qt) u_{k+M+1}(qt) u_{k+2M+2}(qt)$. 
Those of $W(qt) U(t)$ are, respectively, $w_{k-M-1}(qt)u_{k}(qt)$ and $w_{k+M+1}(qt)u_{k+2M+2}(qt) u_{k}(t)$. 
Using the third equation in \eqref{eqn:extended_qLV}, we see that the $(k+M+1,k)$ and $(k+2M+2,k)$ entries of $U(qt)W(qt)-W(qt)U(t)$ are $w_k(qt)(u_{k+M+1}(qt)-u_{k-M-1}(t))$ and $0$, respectively. 
The other entries of $U(qt)W(qt)$ and $W(qt)U(t)$ are all $0$.
Thus we have \eqref{eqn:Lax_extended_qLV}. 
\end{proof}
Moreover, by preparing $2m\times2m$ matrices involving the extended $q$-LV variable $u_{k,s}(t)\coloneqq u_{M_k+s}(t)$ and $w_{k,s}(t)\coloneqq w_{M_k+s}(t)$: 
\[\begin{aligned}
& {\cal U}_s(t)\coloneqq\left(\begin{array}{cccc}
0 & 1 & & \\
u_{1,s}(t) & 0 & \ddots & \\
 & \ddots & \ddots & 1 \\
 & & u_{2m-1,s}(t) & 0 
\end{array}\right), \\
& {\cal W}_s(qt)\coloneqq\left( \begin{array}{ccccc}
0 & & & & \\
0 & \ddots & & & \\
w_{1,s}(t) u_{2,s}(t) & \ddots & \ddots & & \\
 & \ddots & \ddots & \ddots & \\
 & & w_{2m-2,s}(t) u_{2m-1,s}(t) & 0 & 0 
\end{array}\right), 
\end{aligned}\]
we obtain another Lax representation for the extended $q$-LV system \eqref{eqn:extended_qLV}. 
\begin{proposition}
\label{prop:Lax_extended_qLV_1}
For the block diagonal matrices ${\cal U}(t)\coloneqq {\rm diag}({\cal U}_0(t),{\cal U}_1(t),\dots,$ ${\cal U}_m(t))$ and ${\cal W}(t)\coloneqq {\rm diag}({\cal W}_{0}(t),{\cal W}_1(t),\dots,{\cal W}_m(t))$, it holds that
\begin{align}
\label{eqn:Lax_extended_qLV_1_all}
D_q{\cal U}(t)={\cal U}(qt){\cal W}(qt)-{\cal W}(qt){\cal U}(t). 
\end{align}
\end{proposition}
\begin{proof}
The proof is similar to that of Proposition \ref{prop:Lax_extended_qToda_1}. 
Taking into account that$D_q{\cal U}(t)$ and ${\cal U}(qt){\cal W}(qt)-{\cal W}(qt){\cal U}(t)$ are both block diagonal matrices and focusing on the $s$th diagonal blocks of the block diagonal matrices $D_q{\cal U}(t)$ and ${\cal U}(qt){\cal W}(qt)-{\cal W}(qt){\cal U}(t)$, we find directly that:
\begin{align}
\label{eqn:Lax_extended_qLV_1_sub}
D_q{\cal U}_s(t)={\cal U}_s(qt){\cal W}_s(qt)-{\cal W}_s(qt){\cal U}_s(t),\quad s=0,1,\dots,M.
\end{align}
Noting that ${\cal U}_s(t)$ and ${\cal W}_s(t)$ have the same form as $U(t)$ and $W(t)$ 
in Proposition \ref{prop:Lax_extended_qLV}, respectively, 
we obtain equalities of nonzero entries in \eqref{eqn:Lax_extended_qLV_1_sub}:
\[D_qu_{k,s}(t)=w_{k,s}(qt)\left(u_{k+1,s}(qt)-u_{k-1,s}(t)\right),\]
which is equivalent to the extended $q$-LV system \eqref{eqn:extended_qLV} under the replacements
$u_{k,s}(t)=u_{M_k+s}(t)$ and $w_{k,s}(t)=w_{M_k+s}(t)$. 
Thus, we have \eqref{eqn:Lax_extended_qLV_1_all}. 
\end{proof}
Comparing Propositions \ref{prop:Lax_extended_qToda} and \ref{prop:Lax_extended_qToda_1} 
with Propositions \ref{prop:Lax_extended_qLV} and \ref{prop:Lax_extended_qLV_1}, 
we also find that the extended $q$-Toda equation \eqref{eqn:extended_qToda} and the $q$-LV system \eqref{eqn:extended_qLV} 
both generate similarity transformations of $(M+1)$-tridiagonal matrices and block diagonal matrices. 
This implies that the extended $q$-LV system \eqref{eqn:extended_qLV} has a close relationship 
to the extended $q$-Toda equation \eqref{eqn:extended_qToda}. 
The following proposition actually gives the B\"acklund transformation between the extended $q$-Toda equation \eqref{eqn:extended_qToda} 
and the extended $q$-LV system \eqref{eqn:extended_qLV}. 
\begin{proposition}
\label{prop:Backlund}
The extended $q$-Toda variables $x_{k,s}(t)$, $y_{k,s}(t)$, and $g_{k,s}(t)$
and the extended $q$-LV variables $u_{2k-1,s}(t)$, $u_{2k,s}(t)$, and $w_{2k-1,s}(qt)$ satisfy:
\begin{align}
\label{eqn:Backlund}
\left\{\begin{aligned}
& x_{k,s}(t)=u_{2k-1,s}(t)+u_{2k-2,s}(t),\quad k=1,2,\dots,m,\quad s=0,1,\dots,M,\\ 
& y_{k,s}(t)=u_{2k,s}(t)u_{2k-1,s}(t),\quad k=1,2,\dots,m, \quad s=0,1,\dots,M, 
\end{aligned}\right.
\end{align}
and:
\begin{align}
\label{eqn:Backlund_auxiliary}
g_{k,s}(qt)=w_{2k-1,s}(qt)u_{2k,s}(qt),\quad k=1,2,\dots,m,\quad s=0,1,\dots,M. 
\end{align}
\end{proposition}
\begin{proof}
Using \eqref{eqn:extended_qLV_w} and \eqref{eqn:Backlund}, we can easily rewrite the variables $g_{k,s}(qt)$ 
in \eqref{eqn:extended_qToda_g} as \eqref{eqn:Backlund_auxiliary}. 
Considering \eqref{eqn:Backlund} and \eqref{eqn:Backlund_auxiliary} 
in the extended $q$-Toda equation \eqref{eqn:extended_qToda}, we obtain:
\[\left\{\begin{aligned}
& D_qu_{2k-1,s}(t)+D_qu_{2k-2,s}(t)=w_{2k-1,s}(qt)u_{2k,s}(qt)-w_{2k-3,s}(qt)u_{2k-2,s}(qt),\\
& \qquad k=1,2,\dots,m, \quad s=0,1,\dots,M,\\
& (D_qu_{2k,s}(t))u_{2k-1,s}(t)+u_{2k,s}(qt)D_qu_{2k-1,s}(t)\\
& \quad =w_{2k-1,s}(qt)u_{2k,s}(qt)(u_{2k+1,s}(qt)+u_{2k,s}(qt)-u_{2k-1,s}(t)-u_{2k-2,s}(t)),\\
& \qquad k=1,2,\dots,m-1,\quad s=0,1,\dots,M.
\end{aligned}\right.\]
Recalling here that $w_{k+1,s}(qt)u_{k,s}(t)=w_{k,s}(qt)u_{k+1,s}(qt)$, 
we thus derive the extended $q$-LV system \eqref{eqn:extended_qLV}.
\end{proof}
\begin{figure}[tb]
\centering
\includegraphics[width=0.7\textwidth]{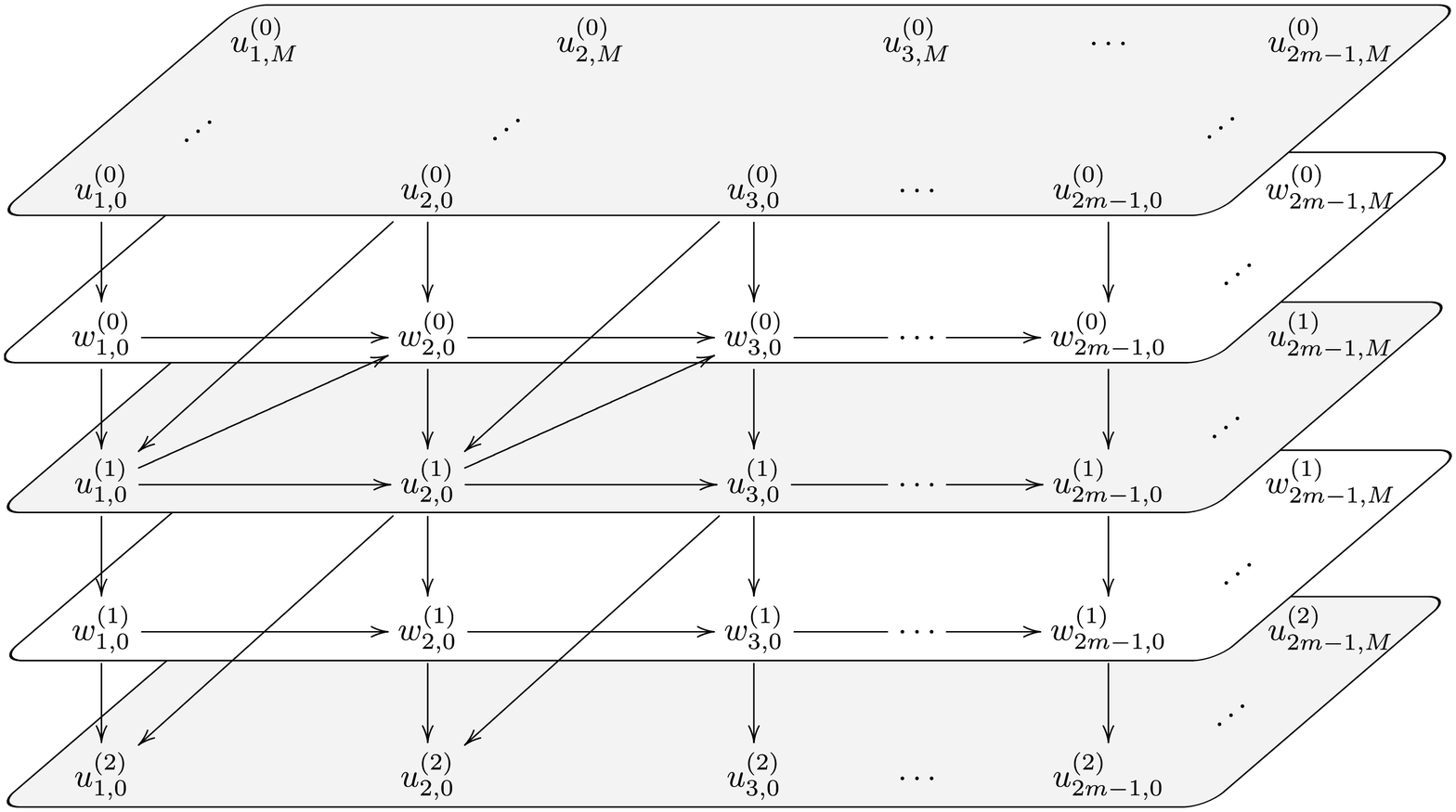}
\caption{Discrete-time evolution in the extended $q$-discrete Lotka-Volterra system \eqref{eqn:extended_q_dLV}. \label{fig:diagram_time_evolution_qdLV}}
\end{figure}
\par
Now we consider a time-discretization of the extended $q$-LV system \eqref{eqn:extended_qLV}. 
Again, by considering a discrete-time sequence $t^{(0)},t^{(1)},\dots$ such that $t^{(n)}\coloneqq q^{-1} t^{(n-1)}$\yadd{,} where $t^{(0)}>0$, 
we can rewrite the extended $q$-LV system \eqref{eqn:extended_qLV} with $t=t^{(n)}$ as:
\begin{align}
\label{eqn:extended_q_dLV}
\left\{\begin{aligned}
& \dfrac{u_{k,s}^{(n+1)}-u_{k,s}^{(n)}}{(1-q)t^{(n)}}=w_{k,s}^{(n)}\left(u_{k+1,s}^{(n)}-u_{k-1,s}^{(n+1)}\right),\quad k=1,2,\dots,2m-1,\\
& u_{0,s}^{(n)}\coloneqq 0,\quad u_{2m,s}^{(n)}\coloneqq 0,
\end{aligned}\right. 
\end{align}
where $u_{k,s}^{(n+1)}\coloneqq u_{k,s}(t^{(n)})$ and $w_{k,s}^{(n+1)}\coloneqq w_{k,s}(t^{(n)})$ and $w_{k,s}^{(n)}$ satisfy:
\begin{align}
\label{eqn:extended_q_dLV_w}
\left\{\begin{aligned}
& w_{1,s}^{(n)}=\dfrac{u_{1,s}^{(n)}}{1+(1-q)t^{(n)}u_{1,s}^{(n)}},\quad
w_{k,s}^{(n)}=\dfrac{u_{k,s}^{(n)}}{u_{k-1,s}^{(n+1)}}w_{k-1,s}^{(n)},\quad k=2,3,\dots,2m-1,\\
& w_{0,s}^{(n)}\coloneqq 0,\quad w_{2m,s}^{(n)}\coloneqq 0.
\end{aligned}\right. 
\end{align}
We can easily check that the sequences $\{u_{k,s}^{(n)}\}_{n=0,1,\dots}$ and $\{w_{k,s}^{(n)}\}_{n=0,1,\dots}$ are uniquely determined under the initial settings of $u_{k,s}^{(n)}$ and $w_{k,s}^{(n)}$. 
Figure \ref{fig:diagram_time_evolution_qdLV} shows the discrete-time evolution from $n$ to $n+1$ in the $q$-discrete LV system \eqref{eqn:extended_q_dLV}. 
Using $U^{(n)}\coloneqq U(t^{(n-1)})$ and $W^{(n)}\coloneqq W(t^{(n-1)})$, we obtain a matrix representation for the extended $q$-discrete LV system \eqref{eqn:extended_q_dLV}.
\begin{theorem}
\label{thm:similar_transformation_dqLV}
For the $(M+1)$-tridiagonal matrix $U^{(n)}$ and $(2M+2)$-subdiagonal matrix $W^{(n)}$, it holds that
\[(I_{M_{2m}+M}+(1-q)t^{(n)}W^{(n)})U^{(n+1)}=U^{(n)}(I_{M_{2m}+M}+(1-q)t^{(n)}W^{(n)}).\]
\end{theorem}
\noindent
Since $\det(I_{M_{2m}+M}+(1-q)t^{(n)}W^{(n)})=1$, it follows from Theorem \ref{thm:similar_transformation_dqLV} that 
$U^{(n+1)}=(I_{M_{2m}+M}+(1-q)t^{(n)}W^{(n)})^{-1}U^{(n)}(I_{M_{2m}+M}+(1-q)t^{(n)}W^{(n)})$. 
Theorem \ref{thm:similar_transformation_dqLV} thus suggests that the extended $q$-discrete LV system \eqref{eqn:extended_q_dLV} 
generates a similarity transformation from $U^{(n)}$ to $U^{(n+1)}$. 
The replacements $x_{k,s}^{(n)}=x_{k,s}(t^{(n-1)})$, $y_{k,s}^{(n)}=y_{k,s}(t^{(n-1)})$, $g_{k,s}^{(n)}=g_{k,s}(t^{(n-1)})$, 
$u_{2k-1,s}^{(n)}=u_{2k-1,s}(t^{(n-1)})$, $u_{2k,s}^{(n)}=u_{2k,s}(t^{(n-1)})$, and $w_{2k-1,s}^{(n)}=w_{2k-1,s}(t^{(n-1)})$ 
in Proposition \ref{prop:Backlund} yields the B\"acklund transformation 
between the extended $q$-discrete Toda equation \eqref{eqn:extended_qdToda} 
and the extended $q$-discrete LV system \eqref{eqn:extended_q_dLV}. 
\begin{proposition}
\label{prop:discrete_Backlund}
The extended $q$-discrete Toda variables $x_{k,s}^{(n)}$, $y_{k,s}^{(n)}$, and $g_{k,s}^{(n)}$
and the extended $q$-discrete LV variables $u_{2k-1,s}^{(n)}$, $u_{2k,s}^{(n)}$, and $w_{2k-1,s}^{(n)}$ satisfy:
\begin{align}
\label{eqn:discrete_Backlund}
\left\{\begin{aligned}
& x_{k,s}^{(n)}=u_{2k-1,s}^{(n)}+u_{2k-2,s}^{(n)},\quad k=1,2,\dots,m,\quad s=0,1,\dots,M, \\ 
& y_{k,s}^{(n)}=u_{2k,s}^{(n)}u_{2k-1,s}^{(n)},\quad k=1,2,\dots,m,\quad s=0,1,\dots,M, 
\end{aligned}\right.
\end{align}
and:
\begin{align}
\label{eqn:discrete_Backlund_auxiliary}
g_{k,s}^{(n)}=w_{2k-1,s}^{(n)}u_{2k,s}^{(n)},\quad k=1,2,\dots,m,\quad s=0,1,\dots,M. 
\end{align}
\end{proposition}
\noindent
With the help of Proposition \ref{prop:discrete_Backlund}, 
we can express $Q_{k,s}^{(n)}$ and $E_{k,s}^{(n)}$ appearing in Proposition \ref{prop:auxiliary_recursion_formula_qdToda}
using the extended $q$-discrete LV variables $u_{2k-1,s}^{(n)}$, $u_{2k,s}^{(n)}$, and $w_{2k-1,s}^{(n)}$ as:
\begin{align}
\label{eqn:auxiliary_variable_extended_qLV}
\left\{\begin{aligned}
& Q_{k,s}^{(n)}\coloneqq u_{2k-1,s}^{(n)}+u_{2k-2,s}^{(n)}+\dfrac{1}{(1-q)t^{(n)}}-(1-q)t^{(n)}w_{2k-3,s}^{(n)}u_{2k-2,s}^{(n)},\\
& \qquad k = 1,2,\dots,m,\quad s=0,1,\dots,M,\\
& E_{k,s}^{(n)}\coloneqq (1-q)t^{(n)}w_{2k-1,s}^{(n)}u_{2k,s}^{(n)},\quad k=1,2,\dots,m,\quad s=0,1,\dots,M. 
\end{aligned}\right.
\end{align}
This implies that the extended $q$-discrete LV system \eqref{eqn:extended_q_dLV} implicitly gives a sequence 
of shifted $LR$ transformations which are the same as shown in Theorem \ref{thm:LR_transformation}.
%
\section{Asymptotic convergence}
\label{sec:5}
In this section, we describe the properties of the Hankel determinants related to an infinite moment sequence, 
and then clarify the determinantal solution expressed using the Hankel determinants 
to the extended $q$-discrete Toda equation \eqref{eqn:extended_qdToda}. 
Moreover, by considering asymptotic expansions of the Hankel determinants as $n\to\infty$, 
we clarify asymptotic convergence as $n\to\infty$ in the extended $q$-discrete Toda equation \eqref{eqn:extended_qdToda}.
\par
We consider monic polynomials with respect to $z$ having distinct roots $\lambda_{1,s},\lambda_{2,s},$ $\dots,\lambda_{m,s}$:
\begin{align}
\label{eqn:polynomial}
p_s(z)\coloneqq (z-\lambda_{1,s})(z-\lambda_{2,s})\cdots (z-\lambda_{m,s}),\quad s=0,1,\dots,M. 
\end{align}
Equation \eqref{eqn:polynomial} can be expanded as:
\begin{equation}
\label{eqn:polynomial_extended}
p_s(z)=z^m+a_{1,s}z^{m-1}+\cdots +a_{m-1,s}z+a_{m,s},\quad s=0,1,\dots,M,
\end{equation}
where $a_{1,s},a_{2,s},\dots,a_{m,s}$ are given in terms of $\lambda_{1,s},\lambda_{2,s},\dots,\lambda_{m,s}$.
Now, we introduce an infinite sequence $\{ f_{k,s}^{(n)} \}_{n=0}^{\infty}$
associated with the constants $\lambda_{1,s},\lambda_{2,s},\dots,\lambda_{m,s}$.
The sequence $\{f_{k,s}^{(n)}\}_{n=0}^{\infty}$ is sometimes called a moment sequence.
Let us assume that $\{ f_{k,s}^{(n)} \}_{n=0}^{\infty}$ satisfies the linear equations involving $a_{1,s},a_{2,s},\dots,a_{m,s}$:
\begin{align}
\label{eqn:moment_linear_equation}
f_{k+m,s}^{(n)}+\sum_{i=1}^ma_{i,s}f_{k+m-i,s}^{(n)}=0,\quad k=0,1,\dots,m,\quad s=0,1,\dots,M,\quad n=0,1,\dots,
\end{align}
and discrete-time evolutions from $n$ to $n+1$:
\begin{align}
\label{eqn:moment_time_evolution}
f_{k,s}^{(n+1)}=f_{k+1,s}^{(n)}-\mu^{(n)}f_{k,s}^{(n)},\quad k=0,1,\dots,m-1,\quad s=0,1,\dots,M, 
\end{align}
where $f_{2m+1,s}^{(n)}\coloneqq 0$ and $\mu^{(n)}$ are arbitrary constants. 
It is obvious that the moment sequence $\{f_{k,s}^{(n)}\}_{n=0}^{\infty}$ is uniquely determined 
from the values of $f_{0,s}^{(0)},f_{1,s}^{(0)},\dots,f_{m-1,s}^{(0)}$ for each $s$. 
It is worth noting that, for each $s$, the moment sequence $\{f_{k,s}^{(n)}\}_{n=0}^{\infty}$ is equivalent to 
the moment sequence $\{f_{k}^{(n)}\}_{k,n=0}^{\infty}$ given by \cite{Shinjo_2018}. 
Thus, by extending their discussion, we can easily derive a proposition for the moments $f_{k,s}^{(n)}$. 
\begin{proposition}(cf.~\cite{Shinjo_2018})
For each $s$, the moments $f_{k,s}^{(n)}$ can be expressed as:
\begin{align}
\label{eqn:moment_expressed}
f_{k,s}^{(n)}=\sum_{\ell=1}^mc_{\ell,s}\lambda_{\ell,s}^k\rho_{\ell,s}^{(n)}\quad k=0,1,\dots,2m,\quad n=0,1,\dots, 
\end{align}
where:
\begin{align}
\label{eqn:rho}
\rho_{\ell,s}^{(0)}\coloneqq 1,\quad\rho_{\ell,s}^{(n)}\coloneqq\prod_{j=0}^{n-1}(\lambda_{\ell,s}-\mu^{(j)}),\quad\ell =1,2,\dots,m, 
\end{align}
and constants $c_{1,s},c_{2,s},\dots,c_{m,s}$ satisfy:
\begin{align}
\label{eqn:c_condition}
\left(\begin{array}{c}
c_{1,s} \\ c_{2,s} \\ \vdots \\ c_{m,s} 
\end{array}\right) 
=\left(\begin{array}{cccc}
1 & 1 & \cdots & 1 \\
\lambda_{1,s} & \lambda_{2,s} & \cdots & \lambda_{m,s}\\
\vdots & \vdots & & \vdots\\
\lambda_{1,s}^{m-1} & \lambda_{2,s}^{m-1} & \cdots & \lambda_{m,s}^{m-1}
\end{array}\right)^{-1} 
\left(\begin{array}{c}
f_{0,s}^{(0)} \\ f_{1,s}^{(0)} \\ \vdots \\ f_{m-1,s}^{(0)} 
\end{array}\right). 
\end{align}
\end{proposition}
Similarly to Ref.\cite{Shinjo_2018}, by considering Hankel determinants of degree $k$ given using the moment $f_{k,s}^{(n)}$ as:
\begin{align}
\label{eqn:determinant}
\begin{aligned}
& H_{-1,s}^{(n)}\coloneqq 0,\quad H_{0,s}^{(n)}\coloneqq 1,\\
& H_{k,s}^{(n)}\coloneqq\left\vert \begin{array}{cccc}
 f_{0,s}^{(n)} & f_{1,s}^{(n)} & \cdots & f_{k-1,s}^{(n)}\\
 f_{1,s}^{(n)} & f_{2,s}^{(n)} & \cdots & f_{k,s}^{(n)}\\
 \vdots & \vdots & \ddots & \vdots\\
 f_{k-1,s}^{(n)} & f_{k,s}^{(n)} & \cdots & f_{2k-2,s}^{(n)}
\end{array}\right\vert,\quad k=1,2,\dots,m,\quad n=0,1,\dots, 
\end{aligned}
\end{align}
we also obtain two propositions for the Hankel determinants $H_{k,s}^{(n)}$. 
\begin{proposition}(cf.~\cite{Shinjo_2018})
For each $s$, the Hankel determinants $H_{k,s}^{(n)}$ satisfy:
\begin{align}
\label{prop:determinant_extended}
H_{k,s}^{(n)}=
\left\{
\begin{aligned}
&\sum_{1\leq i_1<i_2<\cdots <i_k\leq m} { \cal C}_{i_1,i_2,\dots,i_k;s}\rho_{i_1,s}^{(n)}\rho_{i_2,s}^{(n)}\cdots\rho_{i_k,s}^{(n)},\quad k=1,2,\dots,m,\\
&0,\quad k=m+1,\\
\end{aligned}
\right.
\end{align}
where:
\begin{align}
\label{eqn:determinant_extended_C}
{\cal C}_{i_1,i_2,\dots,i_k;s}\coloneqq\left\vert\begin{array}{cccc}
1 & 1 & \cdots & 1 \\
\lambda_{i_1,s} & \lambda_{i_2,s} & \cdots & \lambda_{i_k,s} \\
\vdots & \vdots & \ddots & \vdots \\
\lambda_{i_1,s}^{k-1} & \lambda_{i_1,s}^{k-1} & \cdots & \lambda_{i_k,s}^{k-1}
\end{array}\right\vert\left\vert\begin{array}{cccc}
c_{i_1,s} & c_{i_1,s} \lambda_{i_1,s} & \cdots & c_{i_1,s} \lambda_{i_1,s}^{k-1} \\
c_{i_2,s} & c_{i_2,s} \lambda_{i_2,s} & \cdots & c_{i_2,s} \lambda_{i_2,s}^{k-1} \\
\vdots & \vdots & \ddots & \vdots \\
c_{i_k,s} & c_{i_k,s} \lambda_{i_k,s} & \cdots & c_{i_k,s} \lambda_{i_k,s}^{k-1}
\end{array}\right\vert.
\end{align}
\end{proposition}
Let us assume that $H_{k,s}^{(n)}\neq 0$ for $k=1,2,\dots,m$ and $s=0,1,\dots,M$. 
Let us define the $k$th degree polynomials with respect to $z$ as:
\begin{align}
\label{eqn:orthogonal_polynomial}
{\cal H}_{k,s}^{(n)}(z)=\dfrac{H_{k,s}^{(n)}(z)}{H_{k,s}^{(n)}},\quad k=1,2,\dots,m, 
\end{align}
where
\begin{align}
\label{eqn:determinant_func}
\begin{aligned}
& H_{-1,s}^{(n)}(z)\coloneqq 0,\quad H_{0,s}^{(n)}(z)\coloneqq 1,\\
& H_{k,s}^{(n)}(z)\coloneqq\left\vert\begin{array}{cccc}
f_{0,s}^{(n)} & f_{1,s}^{(n)} & \cdots & f_{k,s}^{(n)} \\
f_{1,s}^{(n)} & f_{2,s}^{(n)} & \cdots & f_{k+1,s}^{(n)} \\
\vdots & \vdots & \ddots & \vdots \\
f_{k-1,s}^{(n)} & f_{k,s}^{(n)} & \cdots & f_{2k-1,s}^{(n)} \\
1 & z & \cdots & z^k
\end{array}\right\vert,\quad k=1,2,\dots,m,\quad n=0,1,\dots. 
\end{aligned}
\end{align}
\noindent
With the help of \yadd{\cite{Shinjo_2018}}, we immediately derive two lemmas concerning the polynomials ${\cal H}_{k,s}^{(n)}(z)$.
\begin{lemma}(cf.~\cite{Shinjo_2018})
Let us assume that $H_{k,s}^{(n)}\neq 0$ and $\mu^{(n)}=-1/ [(1-q)t^{(n)}]$. 
For each $s$, the Hadamard polynomials ${\cal H}_{k,s}^{(n)}$ satisfy:
\begin{align}
\left\{\begin{aligned}
& (z-\mu^{(n)}){\cal H}_{k-1,s}^{(n+1)}(z)={\cal H}_{k,s}^{(n)}(z)+Q_{k,s}^{(n)}{\cal H}_{k-1,s}^{(n)}(z),\quad k=1,2,\dots,m,\\
& {\cal H}_{k,s}^{(n)}(z)={\cal H}_{k,s}^{(n+1)}(z)+E_{k,s}^{(n)}{\cal H}_{k-1,s}^{(n+1)}(z),\quad k=0,1,\dots,m, 
\end{aligned}\right.
\end{align}
where
\begin{align}
\left\{\begin{aligned}
\label{eqn:auxiliary_variable_determinantal_solution}
& Q_{k,s}^{(n)}=\dfrac{H_{k,s}^{(n+1)}H_{k-1,s}^{(n)}}{H_{k,s}^{(n)}H_{k-1,s}^{(n+1)}},\quad k=1,2,\dots,m,\quad s=0,1,\dots,M,\\
& E_{k,s}^{(n)}=\dfrac{H_{k+1,s}^{(n)}H_{k-1,s}^{(n+1)}}{H_{k,s}^{(n)}H_{k,s}^{(n+1)}},\quad k=0,1,\dots,m,\quad s=0,1,\dots,M.
\end{aligned}\right.
\end{align}
Moreover, ${\cal H}_{m,s}^{(n)}(z)$ are just $p_s(z)$, which are characteristic polynomials of matrices with eigenvalues $\lambda_{1,s},\lambda_{2,s},\ldots,\lambda_{m,s}$.
\end{lemma}
Combining \eqref{eqn:auxiliary_variable_determinantal_solution} with \eqref{eqn:auxiliary_recursion_formula_qdToda_Q} and \eqref{eqn:auxiliary_recursion_formula_qdToda_E}, we thus obtain the determinantal solution to the extended $q$-discrete Toda equation \eqref{eqn:extended_qdToda}.
\begin{theorem}
\label{thm:solutions_extended_qdToda}
The extended $q$-discrete Toda variables $x_{k,s}^{(n)}$ and $y_{k,s}^{(n)}$ can be expressed using the Hankel determinants $H_{k,s}^{(n)}$ as
\begin{align}
\label{eqn:qdToda_solution}
\begin{aligned}
& x_{k,s}^{(n)}=\dfrac{H_{k,s}^{(n+1)}H_{k-1,s}^{(n)}}{H_{k,s}^{(n)}H_{k-1,s}^{(n+1)}}
+\dfrac{H_{k,s}^{(n)}H_{k-2,s}^{(n+1)}}{H_{k-1,s}^{(n)}H_{k-1,s}^{(n+1)}}-\dfrac{1}{(1-q)t^{(n)}},\quad k=1,2,\dots,m,\\
& y_{k,s}^{(n)}=\dfrac{H_{k-1,s}^{(n)}H_{k+1,s}^{(n)}}{(H_{k,s}^{(n)})^2},\quad k=0,1,\dots,m.
\end{aligned}
\end{align}
\end{theorem}
\noindent
According to \yadd{\cite{Shinjo_2018}}, asymptotic expansions as $n\to\infty$ of the Hankel determinants are given as the following lemma.
\begin{lemma}(cf.~\cite{Shinjo_2018})
\label{lemma:asymptotic_behavior}
Let us assume that $\vert\lambda_{1,s}-\mu^{(n)}\vert >\vert\lambda_{2,s}-\mu^{(n)}\vert >\cdots >\vert\lambda_{m,s}-\mu^{(n)}\vert$, 
and $\varrho_{k,s}$ are constants that satisfy $0\leq\varrho_{k,s}<1$ and 
$\varrho_{k,s}> \vert\lambda_{k+1,s}-\mu^{(n)}\vert /\vert\lambda_{k,s}-\mu^{(n)}\vert$. 
For sufficiently large $n$, it holds that
\[H_{k,s}^{(n)}={\cal C}_{1,2,\dots,k;s}\rho_{1,s}^{(n)}\rho_{2,s}^{(n)}\cdots\rho_{k,s}^{(n)}\left(1+O(\varrho_{k,s}^{n})\right).\]
\end{lemma}
From Theorem \ref{thm:solutions_extended_qdToda} and Lemma \ref{lemma:asymptotic_behavior}, 
we therefore have the convergence theorem in the extended $q$-Toda equation \eqref{eqn:extended_qdToda}. 
\begin{theorem}
For each $s$, let us assume that the moment sequence $\{f_{k,s}^{(n)}\}_{n}^{\infty}$ satisfies $H_{k,s}^{(n)}\neq0$, 
and $\vert\lambda_{1,s}-\mu^{(n)}\vert >\vert\lambda_{2,s}-\mu^{(n)}\vert >\cdots >\vert\lambda_{m,s}-\mu^{(n)}\vert$. 
Then, it holds that
\begin{align}
\label{eqn:qdToda_solution_x_converge}
& \lim_{n\to\infty}x_{k,s}^{(n)}=\lambda_{k,s},\quad k=1,2,\dots,m,\\
\label{eqn:qdToda_solution_y_converge}
& \lim_{n\to\infty}y_{k,s}^{(n)}= 0,\quad k=1,2,\dots,m-1.
\end{align}
\end{theorem}
\begin{proof}
Using Lemma \ref{lemma:asymptotic_behavior} and 
considering $1/[(1-q)t^{(n)}]\to 0$ as $n\to\infty$ in Theorem \ref{thm:solutions_extended_qdToda}, we obtain
\begin{align*}
x_{k,s}^{(n)} & =\dfrac{\rho_{k,s}^{(n+1)}}{\rho_{k,s}^{(n)}}
\dfrac{\left(1+O(\varrho_{k,s}^{n+1})\right)\left(1+O(\varrho_{k-1,s}^{n})\right)}
{\left(1+O(\varrho_{k,s}^{n})\right)\left(1+O(\varrho_{k-1,s}^{n+1})\right)}\\
& \quad +\dfrac{{\cal C}_{1,2,\dots,k;s}{\cal C}_{1,2,\dots,k-2;s}}{{\cal C}_{1,2,\dots,k-1;s}^2}\dfrac{\rho_{k,s}^{(n)}}{\rho_{k-1,s}^{(n+1)}}
\dfrac{\left(1+O(\varrho_{k,s}^{n})\right)\left(1+O(\varrho_{k-2,s}^{n+1})\right)}
{\left(1+O(\varrho_{k-1,s}^{n})\right)\left(1+O(\varrho_{k-1,s}^{n+1})\right)}, \\
y_{k,s}^{(n)} & =\dfrac{{\cal C}_{1,2,\dots,k+1;s}{\cal C}_{1,2,\dots,k-1;s}}{{\cal C}_{1,2,\dots,k;s}^2}
\dfrac{\rho_{k+1,s}^{(n)}}{\rho_{k,s}^{(n)}}
\dfrac{\left(1+O(\varrho_{k+1,s}^{n})\right)\left(1+O(\varrho_{k-1,s}^{n})\right)}{\left(1+O(\varrho_{k,s}^{n})\right)^2}. 
\end{align*}
Thus, we derive \eqref{eqn:qdToda_solution_x_converge} and \eqref{eqn:qdToda_solution_y_converge}. 
\end{proof}
Considering $x_{k,s}^{(n)}\to\lambda_{k,s}$ and $y_{k,s}^{(n)}\to0$ as $n\to\infty$ in Proposition \ref{prop:discrete_Backlund}, we also have a convergence theorem in the extended $q$-discrete LV system \eqref{eqn:extended_q_dLV}.
\begin{theorem}
For each $s$, let us assume that the moment sequence $\{f_{k,s}^{(n)}\}_{n}^{\infty}$ satisfies $H_{k,s}^{(n)}\neq0$, and $\vert\lambda_{1,s}-\mu^{(n)}\vert >\vert\lambda_{2,s}-\mu^{(n)}\vert >\cdots >\vert\lambda_{m,s}-\mu^{(n)}\vert$. 
Then, it holds that
\begin{align*}
& \lim_{n\to\infty}u_{2k-1,s}^{(n)}=\lambda_{k,s},\quad k=1,2,\dots,m,\\
\label{eqn:qdToda_solution_y_converge}
& \lim_{n\to\infty}u_{2k,s}^{(n)}= 0,\quad k=1,2,\dots,m-1.
\end{align*}
\end{theorem}
\par
In the remainder of this section, we present numerical examples to show the convergence to matrix eigenvalues in the extended $q$-discrete Toda equation \eqref{eqn:extended_qdToda}. 
We used floating point arithmetic on a computer with a Windows 10 Professional
operating system with an Intel(R) Core(TM) i5-7200U CPU @ 2.50 GHz 2.71 GHz, and employed the software Python 3.8.2. 
\begin{figure}[tb]
\centering
\includegraphics[width=0.7\textwidth]{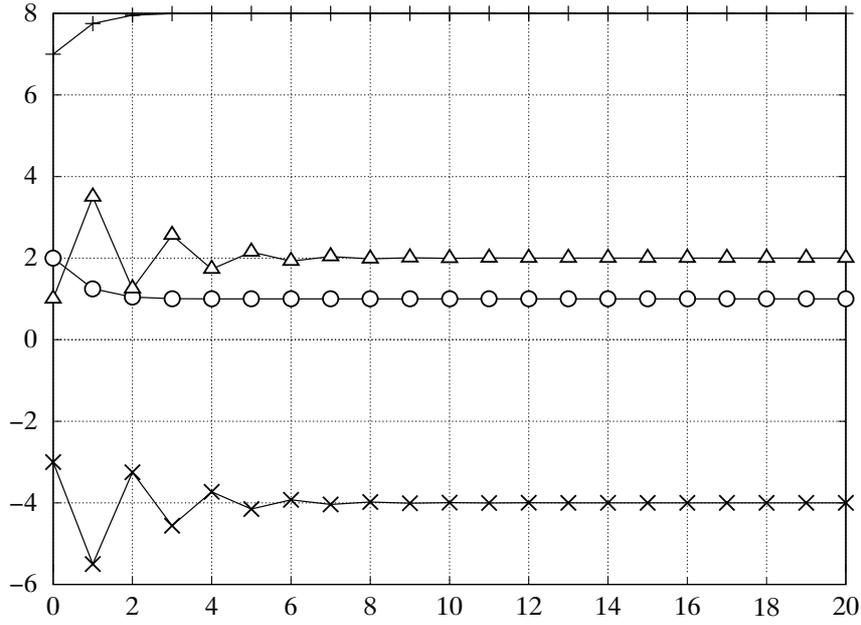}
\caption{Discrete-time $n$ ($x$-axis) versus values of $x_{1,0}^{(n)}$, $x_{1,1}^{(n)}$, $x_{2,0}^{(n)}$, and $x_{2,1}^{(n)}$ ($y$-axis). 
Pluses: $\{x_{1,0}^{(n)}\}_{n=0,1,\dots,20}$; crosses: $\{x_{1,1}^{(n)}\}_{n=0,1,\dots,20}$; circles: $\{x_{2,0}^{(n)}\}_{n=0,1,\dots,20}$; and triangles: $\{x_{2,1}^{(n)}\}_{n=0,1,\dots,20}$.\label{fig:convergence_eigenvalue}}
\end{figure}
\par
We first study the computation of eigenvalues of a $4\times4$ $2$-tridiagonal matrix
\[A^{(0)}=\left(\begin{array}{cccc}
7 & & 1 & \\
 & -3 & & 1 \\
6 & & 2 & \\
 & 5 & & 1
\end{array}\right).\]
It is easy to check that the eigenvalues of $A^{(0)}$ are $8,2,1$ and $-4$ because $A^{(0)}$ is a similar matrix to the block diagonal matrix with tridiagonal blocks
\[{\cal A}^{(0)}=\left(\begin{array}{cccc}
7 & 1 & & \\
6 & 2 & & \\
 & & -3 & 1 \\
 & & 5 & 1
\end{array}\right).\] 
We set the parameters in the extended $q$-discrete Toda equation \eqref{eqn:extended_qdToda} as $m=2$, $M=1$, $t^{(0)}=1$ and $q=1/2$. 
The initial values of the extended $q$-discrete Toda equation \eqref{eqn:extended_qdToda} are directly 
given from entries of $A^{(0)}$ as $x_{1}^{(0)}=x_{1,0}^{(0)}=7$, $x_{2}^{(0)}=x_{1,1}^{(0)}=-3$, $x_{3}^{(0)}=x_{2,0}^{(0)}=2$, 
$x_{4}^{(0)}=x_{2,1}^{(0)}=1$, $y_{1}^{(0)}=y_{1,0}^{(0)}=6$, and $y_{2}^{(0)}=y_{1,1}^{(0)}=5$. 
This differs from the case of the qd recursion formula, namely, the original discrete Toda equation that 
requires the decomposition of the target matrix into a product of lower and upper diagonal matrices. 
Figure \ref{fig:convergence_eigenvalue} draws the approach of the extended $q$-discrete Toda variables
$x_{1,0}^{(n)}, x_{2,0}^{(n)}$, $x_{1,1}^{(n)}$ and $x_{2,1}^{(n)}$ 
to the eigenvalues $8$, $1$, $-4$ and $2$, respectively, as $n$ grows larger. 
\begin{figure}[tb]
\centering
\includegraphics[width=0.7\textwidth]{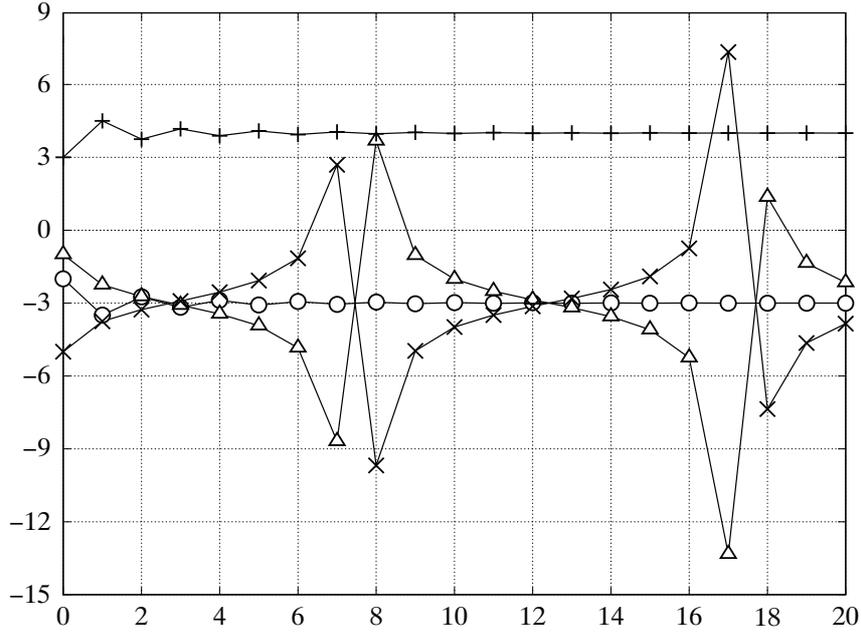}
\caption{Discrete-time $n$ ($x$-axis) versus values of $x_{1,0}^{(n)}$, $x_{1,1}^{(n)}$, $x_{2,0}^{(n)}$, and $x_{2,1}^{(n)}$ ($y$-axis). 
Pluses: $\{x_{1,0}^{(n)}\}_{n=0,1,\dots,20}$; crosses: $\{x_{1,1}^{(n)}\}_{n=0,1,\dots,20}$; 
circles: $\{x_{2,0}^{(n)}\}_{n=0,1,\dots,20}$; and triangles: $\{x_{2,1}^{(n)}\}_{n=0,1,\dots,20}$.\label{fig:not_convergence_eigenvalue}}
\end{figure}
\par
Next, we prepare the $4\times4$ matrices
\[A^{(0)}=\left(\begin{array}{cccc}
3 & & 1 & \\
 & -5 & & 1 \\
6 & & -2 & \\
 & -5 & & -1
\end{array}\right) 
\quad and \quad
{\cal A}^{(0)}=\left(\begin{array}{cccc}
3 & 1 & & \\
6 & -2 & & \\
 & & -5 & 1 \\
 & & -5 & -1
\end{array}\right)\]
with real eigenvalues $4$ and $-3$ and complex eigenvalues $-3+i$ and $-3-i$. 
In the extended $q$-discrete Toda equation \eqref{eqn:extended_qdToda}, we set the initial values and parameters as
$x_{1}^{(0)}=x_{1,0}^{(0)}=3$, $x_{2}^{(0)}=x_{1,1}^{(0)}=-5$, $x_{3}^{(0)}=x_{2,0}^{(0)}=-2$, $x_{4}^{(0)}=x_{2,1}^{(0)}=-1$, 
$y_{1}^{(0)}=y_{1,0}^{(0)}=6$, and $y_{2}^{(0)}=y_{1,1}^{(0)}=-5$, as well as $m=2$, $M=1$, $t^{(0)}=1$ and $q=1/2$. 
Figure \ref{fig:not_convergence_eigenvalue} implies that $x_{1,0}^{(n)}$ and $x_{2,0}^{(n)}$ respectively converge to real eigenvalues $4$ and $-3$,but $x_{2,0}^{(n)}$ and $x_{2,1}^{(n)}$ do not converge as $n\to\infty$. 
In other words, we numerically verified that the extended $q$-discrete Toda equation \eqref{eqn:extended_qdToda} 
can be applied to computing only real eigenvalues even in the case where target matrices have complex eigenvalues. 
Though complex eigenvalues are computed using the extended $q$-discrete Toda equation \eqref{eqn:extended_qdToda}, 
those of the block that is not diagonalized are almost unchanged from $-3+i$ and $-3-i$. 
\begin{figure}[tb]
\centering
\includegraphics[width=0.7\textwidth]{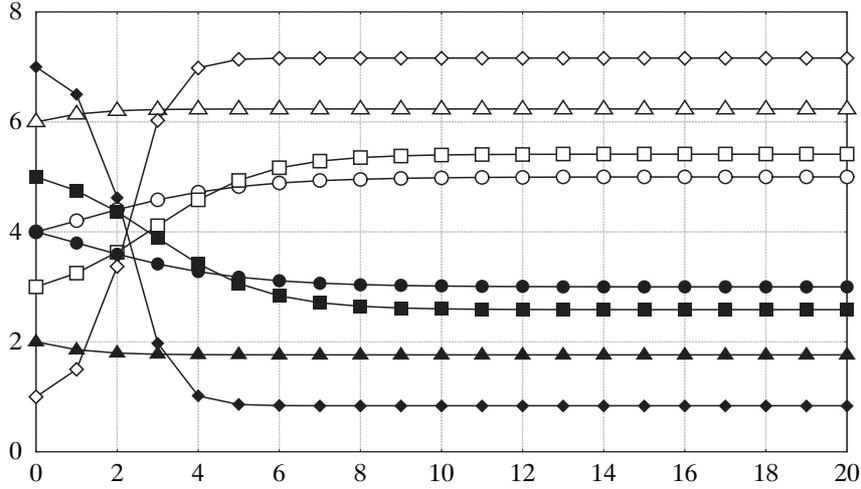}
\caption{Discrete-time $n$ ($x$-axis) versus values of $x_{1,0}^{(n)},x_{1,1}^{(n)},x_{1,2}^{(n)},x_{1,3}^{(n)},x_{2,0}^{(n)},x_{2,1}^{(n)},x_{2,2}^{(n)}$ and $x_{2,3}^{(n)}$ ($y$-axis). 
Circle: $\{x_{1,0}^{(n)}\}_{n=0,1,\dots,20}$; squares: $\{x_{1,1}^{(n)}\}_{n=0,1,\dots,20}$; 
triangles: $\{x_{1,2}^{(n)}\}_{n=0,1,\dots,20}$; diamonds: $\{x_{1,3}^{(n)}\}_{n=0,1,\dots,20}$; 
filled-circle: $\{x_{2,0}^{(n)}\}_{n=0,1,\dots,20}$; filled-squares: $\{x_{2,1}^{(n)}\}_{n=0,1,\dots,20}$; 
filled-triangles: $\{x_{2,2}^{(n)}\}_{n=0,1,\dots,20}$; filled-diamonds: $\{x_{2,3}^{(n)}\}_{n=0,1,\dots,20}$.
\label{fig:M_parameter_eigenvalue}}
\end{figure}
\par
The last example concerns the case of $8\times8$ matrices
\[A^{(0)}=\left(\begin{array}{cccccccc}
4 & & & & 1 \\
& 3 & & & & 1 \\
& & 6 & & & & 1 \\
& & & 1 & & & & 1 \\
1& & & & 4 \\
& 1 & & & & 5 \\
& & 1 & & & & 2 \\
& & & 1 & & & & 7
\end{array}\right) 
\quad and \quad
{\cal A}^{(0)}=\left(\begin{array}{cccccccc}
4 & 1 & & \\
1 & 4 & & \\
 & & 3 & 1 \\
 & & 1 & 5 \\
 & & & & 6 & 1 \\
 & & & & 1 & 2 \\
 & & & & & & 1 & 1 \\
 & & & & & & 1 & 7 \\
\end{array}\right)\]
with eigenvalues $4\pm1,4\pm\sqrt{2},4\pm\sqrt{5}$ and $4\pm\sqrt{10}$. 
We set the initial values and parameters as 
$x_{1}^{(0)}=x_{1,0}^{(0)}=4$, $x_{2}^{(0)}=x_{1,1}^{(0)}=3$, $x_{3}^{(0)}=x_{1,2}^{(0)}=6$, $x_{4}^{(0)}=x_{2,1}^{(0)}=1$, 
$x_{5}^{(0)}=x_{2,0}^{(0)}=4$, $x_{6}^{(0)}=x_{2,1}^{(0)}=5$, $x_{7}^{(0)}=x_{2,2}^{(0)}=2$, $x_{8}^{(0)}=x_{2,1}^{(0)}=7$, 
$y_{1}^{(0)}=y_{1,0}^{(0)}=1$, $y_{2}^{(0)}=y_{1,1}^{(0)}=1$, $y_{3}^{(0)}=y_{1,2}^{(0)}=1$, $y_{4}^{(0)}=y_{1,3}^{(0)}=1$ 
as well as $m=2$, $M=3$, $t^{(0)}=1$ and $q=1/2$. 
in the extended $q$-discrete Toda equation \eqref{eqn:extended_qdToda}. 
Figure \ref{fig:M_parameter_eigenvalue} shows that $x_{1,0}^{(n)},x_{1,1}^{(n)},x_{1,2}^{(n)},x_{1,3}^{(n)},x_{2,0}^{(n)},x_{2,1}^{(n)},x_{2,2}^{(n)}$ and $x_{2,3}^{(n)}$, respectively, \yadd{converge} to eigenvalues $5,4+\sqrt{2},4+\sqrt{5},4+\sqrt{10},3,4-\sqrt{2},4-\sqrt{5}$ and $4-\sqrt{10}$ as $n\to\infty$. Comparing Figure \ref{fig:M_parameter_eigenvalue} with Figure \ref{fig:convergence_eigenvalue}, we can see that $M=3$ case has the asymptotic behavior as $n$ grows larger similarly to the $M=1$ case.
%
%
\section{Concluding remarks}
\label{sec:6}
In this paper, we first considered an extension of the $q$-Toda equation which is a $q$-analogue of the famous Toda equation, and related it to eigenvalue problems of $(M+1)$-tridiagonal matrices and block diagonal matrices whose blocks are tridiagonal. 
We next showed that time-discretization of the extended $q$-Toda equation can generate similarity transformations of the $(M+1)$-tridiagonal matrices. 
We also found the relationship between the Toda case and LV case and related the LV case to the same eigenvalue problem as the Toda case. 
Finally, we proved convergence to eigenvalues in the extended $q$-discrete Toda equation and presented numerical examples for verifying it. 
\par
Our future works \yadd{from the numerical analysis perspective} are to examine the numerical stability of the similarity transformations 
generated by the extended $q$-discrete Toda equation 
and accelerating the convergence speed by introducing explicit shifts, rather than implicit shifts. 
From the viewpoint of the study of integrable systems, 
we plan to relate a $q$-analogue of hungry integrable systems to eigenvalue problems. 
\section*{Acknowledgements}
The authors thank the reviewer for his/her careful reading and constructive suggestions.
This work was partially supported by the joint project of Kyoto University and Toyota Motor Corporation, titled ``Advanced Mathematical Science for Mobility Society''.
\section*{Disclosure statement}
No potential conflict of interest was reported by the authors.
\section*{ORCID}
R. Watanabe: 0000-0002-5758-9587
\bibliographystyle{tfs}
\bibliography{watanabe_qToda.bbl}

\end{document}